\newif\ifreport\reporttrue
\documentclass[conference,letterpaper]{IEEEtran}
\IEEEoverridecommandlockouts
\usepackage{amsmath,amssymb,amsfonts,amsthm}
\usepackage{subcaption}
\usepackage{color}
\usepackage{tabularx}
\usepackage{algorithm}
\usepackage{algpseudocode}

\newtheorem{definition}{Definition}
\newtheorem{lemma}{Lemma}
\newtheorem{proposition}{Proposition}

\def\blue{\color{blue}}

\usepackage{color}
\usepackage{tcolorbox}
\colorlet{blue}{black}
\usepackage{lipsum}
\usepackage{hyperref}
\usepackage[noadjust]{cite}

\begin{document}

\title{On the Monotonicity of Information Aging}
\author{
\IEEEauthorblockN{MD Kamran Chowdhury Shisher \text{and} Yin Sun}
\IEEEauthorblockA{Department of ECE, Auburn University, AL, USA}
\IEEEcompsocitemizethanks{\IEEEcompsocthanksitem{This work was supported in part by NSF grant CNS-2239677 and ARO grant W911NF-21-1-0244.}}
}

\newcommand{\ignore}[1]{{}}
\pagestyle{plain}
\def\blue{\color{blue}}
\maketitle

\begin{abstract}
In this paper, we analyze the monotonicity of information aging in a remote estimation system, where historical observations of a Gaussian autoregressive AR(p) process are used to predict its future values. We consider two widely used loss functions in estimation: (i) logarithmic loss function for maximum likelihood estimation and (ii) quadratic loss function for MMSE estimation. The estimation error of the AR(p) process is written as a generalized conditional entropy which has closed-form expressions. By using a new information-theoretic tool called \emph{$\epsilon$-Markov chain}, we can evaluate the divergence of the AR(p) process from being a Markov chain. When the divergence $\epsilon$ is large, the estimation error of the AR(p) process can be far from a non-decreasing function of the Age of Information (AoI). Conversely, for small divergence $\epsilon$, the inference error is close to a non-decreasing AoI function. Each observation is a short sequence taken from the AR(p) process. As the observation sequence length increases, the parameter $\epsilon$ progressively reduces to zero, and hence the estimation error becomes a non-decreasing AoI function. These results underscore a connection between the monotonicity of information aging and the divergence of from being a Markov chain.
\end{abstract}


\section{Introduction}
Timely updates of sensor measurements are crucial for real-time state estimation and decision-making in networked controlled and cyber-physical systems, such as UAV navigation, real-time surveillance, factory automation, and weather monitoring systems. To evaluate the timeliness of sensor measurements received from a
remote sensor, the
concept of \emph{Age of Information} (AoI) was introduced in \cite{song1990performance, kaul2012real}. Let $U(t)$ be the generation time of the freshest sensor measurement delivered to the receiver by time $t$. The AoI $\Delta(t)$, as a function of time $t$, is defined as
\begin{align}\label{AoIIntro}
    \Delta(t):=t-U(t),
\end{align}
which is the time difference between the current time $t$ and the generation time $U(t)$ of the most recently delivered sensor data. A smaller AoI indicates the presence of recently generated sensor data at the receiver. There has been a significant research efforts on analyzing and optimizing AoI on communication networks \cite{yates2021age,pappas2022age, kaul2012real,sun2017update,yates2015lazy, kadota2018optimizing, SunSPAWC2018, SunNonlinear2019, chen2021uncertainty, wang2022framework, orneeTON2021, Tripathi2019, klugel2019aoi, bedewy2021optimal, sun2019closed, Kadota2018, ornee2023whittle, pan2023sampling, sun2023age, sun2022age, SunTIT2020, tripathi2021computation, shisher2021age, ShisherMobihoc, shisher2023learning, ShisherToN, ornee2023context, ari2023goal}. 

In this paper, we investigate a remote estimation system where a time-varying target is estimated based on observations collected from a sensor. Due to communication delays and transmission errors, the observations delivered at the receiver may not be fresh. Previous studies assumed that system performance degrades monotonically as observations become stale \cite{kadota2018optimizing, SunNonlinear2019, bedewy2021optimal, Tripathi2019, tripathi2021computation}. This assumption was justified for Markov sources \cite{SunNonlinear2019}: However, recent machine learning experimental studies \cite{shisher2021age, ShisherMobihoc, ShisherToN} showed that this monotonic assumption does not always hold.
In certain scenarios, it was found that stale data with AoI $>0$ can even achieve a smaller inference error than fresh data with AoI $= 0$, which is counter-intuitive. Information-theoretic tools was developed in  \cite{shisher2021age, ShisherMobihoc, ShisherToN} to interpret such non-monotonic information aging phenomena in machine learning experiments. To further understand why and in what scenarios information aging could be non-monotonic, in this paper we use a model-based approach to analyze information aging. Specifically, we derive closed-form expressions for the remote estimation error of Gaussian autoregressive AR(p) processes and study how the monotonicity of information aging is affected by the parameters of the AR(p) process. The contributions of this paper are summarized as follows:

\begin{itemize}
   \item We analyze the impact of fresh observations on the remote estimation of a p-th order Gaussian autoregressive process (AR(p)). The AR(p) process is widely used in modeling channel state information \cite{jakes1994microwave}, economic forecasting \cite{stock2001vector}, biomedical signals \cite{isaksson1981computer}, and control systems \cite{klugel2019aoi, champati2019performance, ayan2019age}. Our study is more general than the earlier model-based studies in AoI literature \cite{klugel2019aoi, champati2019performance, ayan2019age}, which are centered on AR(1) processes.
    
    \item The estimation error of the AR(p) process is formulated as a generalized conditional entropy (refer to Lemma \ref{lemma1}). Closed-form expressions are provided for computing the estimation error (see Propositions \ref{prop1}-\ref{prop2}). These expressions are provided for two commonly used loss functions in machine learning and remote estimation: (i) quadratic loss and (ii) logarithmic loss. 

    \item By using a new information-theoretic tool called $\epsilon$-\emph{Markov chain} \cite{shisher2021age, ShisherMobihoc, shisher2023learning}, we evaluate the divergence of the AR(p) process from being Markovian. We then characterize the monotonicity of the estimation error with respect to AoI using the parameter $\epsilon$ (refer to Lemma \ref{lemma2}).  Specifically, if $\epsilon$ is close to zero, the target process is close to being Markov and the estimation error becomes a non-decreasing function of AoI; otherwise, if $\epsilon$ deviates significantly from zero, the target process is far from being Markov and the estimation error can exhibit highly non-monotonic behavior in AoI. 
    
    \item A closed-form expression is provided to compute $\epsilon$ from AR(p) process (see Proposition \ref{prop3}(a)). Additionally, we characterize the parameter $\epsilon$ as a function of the observation time-sequence length. As the observation time-sequence length of the AR(p) process increases to $p$, $\epsilon$ reduces to zero, and hence the estimation error becomes a non-decreasing function of AoI (See Proposition \ref{prop3}(b)). 

    \item Numerical results verify our theoretical findings (see Fig. \ref{fig:AR} and Table \ref{table:1}). 
\end{itemize}
\section{System Model}

Consider the remote estimation system composed of a sensor, a transmitter, and an estimator, as illustrated in Fig. 1. The goal of the system is to estimate a time-varying target $Y_t \in \mathbb R$. We consider that the target $Y_t$ evolves as 
\begin{align} \label{Y}
    Y_t&=X_t+N_t,
\end{align}
where $X_t \in \mathbb R$ follows a discrete-time $p$-th order autoregressive (AR($p$)) linear time-invariant system:
\begin{align}\label{V}
    X_t&=a_1 X_{t-1}+ a_2 X_{t-2}+ \ldots+ a_p X_{t-p}+W_t,
\end{align}
$N_t \in \mathbb R$ and $W_t \in \mathbb R$ are {i.i.d.} Gaussian noises over time with zero mean, and $a_k \in \mathbb R$ for all $k=1, 2, \ldots, p$. Let $\sigma^2_{Y_t}$ and $\sigma^2_{X_t}$ be the variances of $Y_t$ and $X_t$, respectively.

At every time slot $t$, the sensor observes $X_t$ and feeds the observation to the transmitter. The transmitter progressively sends the sensory data to the estimator through a communication channel. Due to communication delays and channel errors, the delivered sensor observations may not be fresh. The most recently received sensor observation at the estimator is $X_{t-\Delta(t)}$ that was generated at time $t-\Delta(t)$. The time difference $\Delta(t) \in \mathbb Z^{+}$ between the generation time $t-\Delta(t)$ and the current time $t$ is the AoI defined in \eqref{AoIIntro}. The estimator takes a consecutive sequence of sensor observations (also called feature sequence) $(X_{t-\Delta(t)}, X_{t-\Delta(t)-1}, \ldots, X_{t-\Delta(t)-l+1}) \in \mathbb R^l$ and the AoI $\Delta(t) \in \mathbb Z^{+}$ as inputs and generates an output $a=\phi(\mathbf X^l_{t-\Delta(t)}, \Delta(t)) \in \mathcal A$, where $\mathbf X^l_{t-\Delta(t)}=[X_{t-\Delta(t)}, X_{t-\Delta(t)-1}, \ldots, X_{t-\Delta(t)-l+1}]$ is the feature sequence vector and the estimator is represented by the function $\phi: \mathbb R^l \times \mathbb Z^{+} \mapsto \mathcal A$. The performance of the estimator is measured by a loss function $L: \mathbb R \times \mathcal A \mapsto \mathbb R$, where $L(y, a)$ is the incurred loss if the output $a \in \mathcal A$ is used for estimation when $Y_t=y$. The loss function $L$ is determined by the \emph{goal} of the remote estimation system. 


We assume that the age process $\{\Delta(t),t = 0,1,2,\ldots\}$ is signal-agnostic and the signal process $\{(Y_t, X_t), t=0, 1, \ldots\}$ is stationary. Under these assumptions, if $\Delta(t)=\delta$, then the minimum estimation error at time slot $t$ can be expressed as a function of AoI $\delta$ and feature sequence length $l$ \cite{ShisherMobihoc, ShisherToN}, given by
\begin{align}\label{instantaneous_err1} 
&\mathrm{err}_{\mathrm{estimation}}(\delta, l)\nonumber\\
&:=\min_{\phi \in \Phi} \mathbb E_{Y, \mathbf X^l \sim P_{Y_t, \mathbf X^{l}_{t-\delta}}}\bigg[L\bigg(Y,\phi\bigg(\mathbf X^l, \delta\bigg)\bigg)\bigg],
\end{align}
where the set of functions $\Phi$ consists of all functions that map from $\mathbb R^l \times \mathbb Z^{+}$ to $\mathcal A$ and $P_{Y_t, \mathbf X^{l}_{t-\delta}}$ is the joint distribution of the target $Y_t$ and the feature $\mathbf X^{l}_{t-\delta}$.
\begin{figure}[t]
\centering
\includegraphics[width=0.50\textwidth]{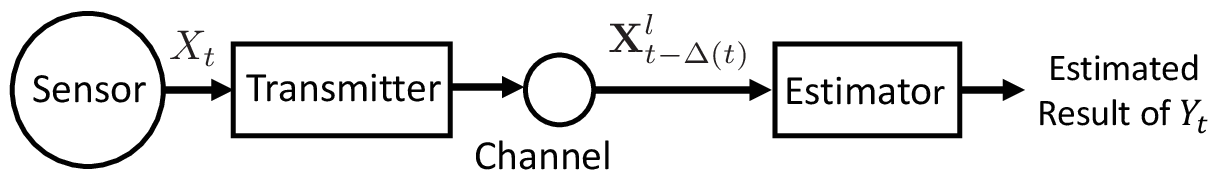}
\caption{\small  A remote estimation system. 
\label{fig:scheduling}
}
\vspace{-3mm}
\end{figure}
\section{Monotonicity of Estimation Error with AoI: An $\epsilon$-Markov Chain Approach}
In this section, we employ an information-theoretic analysis and an $\epsilon$-Markov chain model introduced in \cite{shisher2021age,ShisherMobihoc,ShisherToN} to interpret how the estimation error $\mathrm{err}_{\mathrm{estimation}}(\delta, l)$ varies with the AoI $\delta$ and the feature length $l$.
\subsection{Information-theoretic Metrics for Estimation Error}
We begin with the definitions of $L$-entropy and $L$-conditional entropy. The $L$-entropy of a random variable $Y$ is defined as \cite{ShisherMobihoc, ShisherToN, Dawid1998, farnia2016minimax}
\begin{align}\label{LY}
H_L(Y):=\min_{a \in \mathcal A} \mathbb E_{Y\sim P_Y}[L(Y, a)].
\end{align}
The $L$-conditional entropy of $Y$ given $X$ is defined as \cite{ShisherMobihoc, ShisherToN, Dawid1998, farnia2016minimax}
\begin{align}\label{LYX}
H_L(Y|X):=\mathbb E_{x\sim P_X}[H_L(Y|X=x)],
\end{align}
where $H_L(Y|X=x)$ is given by
\begin{align}\label{LYx}
H_L(Y|X=x)=\min_{a \in \mathcal A} \mathbb E_{Y\sim P_{Y|X=x}}[L(Y, a)].
\end{align}

\begin{lemma}\label{lemma1}
Estimation error $\mathrm{err}_{\mathrm{estimation}}(\delta, l)$ is equal to $L$-conditional entropy of $Y_t$ given $\mathbf X^l_{t-\delta}$, i.e.,
\begin{align}\label{infEstimationError}
\mathrm{err}_{\mathrm{estimation}}(\delta, l)=H_L(Y_t|\mathbf X^l_{t-\delta}).
\end{align}
\end{lemma}
\begin{proof}
See Appendix \ref{plemma1}.
\end{proof}

Lemma \ref{lemma1} implies that we can directly use the $L$-conditional entropy $H_L(Y_t|\mathbf X^l_{t-\delta})$ to analyze the estimation error $\mathrm{err}_{\mathrm{estimation}}(\delta, l)$ of a remote estimation system. By directly using the properties of $L$-information theoretic metrics \cite{ShisherMobihoc, ShisherToN}, the estimation error can be analyzed conveniently. The information-theoretic metrics in the prior studies \cite{soleymani2016optimal, SunSPAWC2018, SunNonlinear2019, wang2022framework,chen2021uncertainty} cannot be directly used to evaluate system performance.

\subsection{Evaluating $L$-conditional Entropy}
We will evaluate the $L$-conditional entropy associated with the loss function $L$. The loss function $L$ is determined based on the objective of a remote estimation system. For example, in minimum mean-squared estimation (MMSE), the loss function is $L_2(y,\hat y) =(y - \hat{y})^2$, where the output $a = \hat {y}$ is an estimate of the target $Y_t =y$. In maximum likelihood estimation of the target distribution, the action $a = Q_{Y_t}$ is a distribution of $Y_t$ and the loss function $L_{\text{log}}(y, Q_{Y_t} ) = - \text{log}~Q_{Y_t} (y)$ is the negative log-likelihood of the target value $Y_t = y$.


\subsubsection{Logarithmic Loss (log loss)}
For log loss $L_{\mathrm{log}}(y, Q_{Y_t}) = -\log Q_{Y_t}(y)$, the $L$-entropy is the differential entropy \cite{cover1999elements, polyanskiy2014lecture}, defined as
\begin{align}
H_{\mathrm{log}}(Y_t) =-\int_{y \in \mathbb R} p_{Y_t}(y)~\mathrm{log}~p_{Y_t}(y)~\mathrm{d}\ \!y,
\end{align}
where $p_{Y_t}$ is the density function of the distribution $P_{Y_t}$ of $Y_t$. Because $Y_t$ is a Gaussian random variable with zero mean, one can obtain \cite{polyanskiy2014lecture}
\begin{align}\label{HYlog}
H_{\mathrm{log}}(Y_t)=\frac{1}{2}\mathrm{log}\left(2\pi e~\mathbb E[Y_t^2]\right).
\end{align}
The $L$-entropy for a discrete random variable associated with log loss is the well known Shannon entropy \cite{ShisherMobihoc, ShisherToN, farnia2016minimax}. The Shannon entropy is always non-negative. However, the differential entropy can be negative, positive, and zero \cite{polyanskiy2014lecture}.
\begin{proposition}\label{prop1}
The $L$-conditional entropy $H_{\mathrm{log}}(Y_t|\mathbf X^l_{t-\delta})$ is given by
 \begin{align}\label{H_log}
   H_{\mathrm{log}}(Y_t|\mathbf X^l_{t-\delta})=&\frac{1}{2}\mathrm{log}\left(\frac{\mathrm{det}(\mathbf R_{[Y_t, \mathbf X^l_{t-\delta}]})}{\mathrm{det}(\mathbf R_{\mathbf X^l_{t}})}\right)\nonumber\\
   &+\frac{1}{2}\mathrm{log} 2\pi e,
   \end{align}
where $\mathrm{det}(\mathbf A)$ denotes the determinant of a square matrix $\mathbf A$, 
\begin{align}\label{autocorr}
\mathbf R_{\mathbf X^l_t}=\mathbb E[(\mathbf X^l_{t})^T \mathbf X^l_{t}]
\end{align}
is an $l\times l$ dimensional auto-correlation matrix of a random vector $\mathbf X^l_{t}$, and   
\begin{align}
\mathbf R_{[Y_t, \mathbf X^l_{t-\delta}]}=\mathbb E\left[[Y_t, \mathbf X^l_{t-\delta}]^T[Y_t, \mathbf X^l_{t-\delta}]\right]
\end{align}
is an $(l+1)\times (l+1)$ dimensional auto-correlation matrix of a random vector $[Y_t,\mathbf X^l_{t-\delta}]=[Y_t, X_{t-\delta}, \ldots, X_{t-\delta-l+1}]$.
\end{proposition} 
\begin{proof}
See Appendix \ref{prop1p}.
\end{proof}

In the special case of feature length $l=1$, from \eqref{H_log}, it can be shown that 
\begin{align}\label{H_log1}
&H_{\mathrm{log}}(Y_t|X_{t-\delta})\nonumber\\
&=\frac{1}{2}\bigg(\mathrm{log}\big(\mathbb E[Y_t^2]-\frac{\mathbb E[X_tX_{t-\delta}]^2}{\mathbb E[X_t^2]}\big)+\mathrm{log}2 \pi e\bigg).
\end{align}

\subsubsection{Quadratic Loss}
For quadratic loss function $L_2(y, \hat{y})= (y - \hat{y})^2$, the $L$-entropy of $Y_t$ is the variance of $Y_t$, given by 
\begin{align}
H_2(Y_t)= \sigma^2_{Y_t}.
\end{align}
Because $\mathbb E[Y_t]=0$, we have 
\begin{align}
H_2(Y_t)=\mathbb E[Y_t^2].
\end{align}
\begin{proposition}\label{prop2}
    The $L$-conditional entropy $H_{2}(Y_t|\mathbf X^l_{t-\delta})$ is given by 
   \begin{align}\label{H_2}
   \!\!\!\!\!\!\! H_{2}(Y_t|\mathbf X^l_{t-\delta})&=\mathbb E[(Y_t-\mathbb E[Y_t|\mathbf X^l_{t-\delta}])^2]\nonumber\\
   &=\mathbb E[Y_t^2]\!-\!\mathbb E[X_t \mathbf X^l_{t-\delta}](\mathbf R_{\mathbf X^l_{t}})^{-1} \mathbb E[X_t \mathbf X^l_{t-\delta}]^T,
   \end{align}
where $\mathbb E[X_t \mathbf X^l_{t-\delta}]=\big[\mathbb E[X_t X_{t-\delta}], \ldots, \mathbb E[X_t X_{t-\delta-l+1}]\big]$ is a $1\times l$ dimensional vector and $\mathbf R_{\mathbf X^l_{t}}$ is an $l\times l$ dimensional auto-correlation matrix of $\mathbf X^l_{t}$ defined in \eqref{autocorr}.
\end{proposition}
\ifreport
\begin{proof}
See Appendix \ref{prop2p}.
\end{proof}
\fi

In the special case of feature length $l=1$, from \eqref{H_2}, it can be shown that 
\begin{align}\label{H_log1}
H_{2}(Y_t|X_{t-\delta})=\mathbb E[Y_t^2]-\frac{\mathbb E[X_t \mathbb X^l_{t-\delta}]^2}{\mathbb E[X_t^2]}.
\end{align}

By utilizing Propositions \ref{prop1}-\ref{prop2}, one can evaluate the $L$-conditional entropy of a data sequence that is generated using a Gaussian AR($p$) system. The $L$-conditional entropy for an AR($4$) model is depicted in Fig. \ref{fig:AR}. The model parameters of the AR($4$) model is presented in Section \ref{simulation}. Fig. \ref{fig:AR} reveals that the $L$-conditional entropy can be a non-monotonic function of AoI $\delta$.

\subsection{$L$-conditional Entropy vs. AoI}
If $Y_t \leftrightarrow \mathbf X^l_{t-\mu}  \leftrightarrow \mathbf X^l_{t-\mu-\nu}$ is a Markov chain for all $\mu,\nu\geq 0$, by the data processing inequality \cite[Lemma 12.1] {Dawid1998}, $H_L(Y_t|\mathbf X^l_{t-\delta})$ is a non-decreasing function of $\delta$. Nevertheless,  the results in Fig. \ref{fig:AR} show that the $L$-conditional entropy is not always a non-decreasing function of $\delta$. This is because $Y_t \leftrightarrow \mathbf X^l_{t-\mu}  \leftrightarrow \mathbf X^l_{t-\mu-\nu}$ is not a Markov chain for all $\mu,\nu\geq 0$, particularly when $l<p$ and $p \geq 2$ in AR($p$) process. 
A relaxation of the data processing inequality is needed to analyze how $H_L(Y_t| \mathbf X^l_{t-\delta})$ varies with $\delta$ for both Markovian and non-Markovian time-series data. To that end, an $\epsilon$-Markov chain model is proposed in \cite{ShisherMobihoc, ShisherToN}. 

\begin{definition}[\textbf{$\epsilon$-Markov Chain}]\label{epsilonMarkovChain}
Given $\epsilon \geq 0$, a sequence of three random variables $Y, X,$ and $Z$ is said to be an \emph{$\epsilon$-Markov chain}, denoted as $Y \overset{\epsilon} \leftrightarrow X \overset{\epsilon} \leftrightarrow Z$, if
\begin{align}\label{epsilon-Markov-def}
\!\!I_{\mathrm{log}}(Y;Z|X)\!= D_{\mathrm{log}}\!\left(P_{Y,X,Z}|| P_{Y|X} P_{Z|X} P_{X} \right)\! \leq \! \epsilon^2,
\end{align}
where 
$I_{\mathrm{log}}(Y;Z|X)$ is Shannon conditional mutual information, 
$D_{\mathrm{log}}\!\left(P_{Y,X,Z}|| P_{Y|X} P_{Z|X} P_{X} \right)$
is KL-divergence between two distributions $P_{Y,X,Z}$ and $P_{Y|X} P_{Z|X} P_{X}$, KL-divergence $D_{\mathrm{log}}(P_Y ||Q_Y)$ between two distributions $P_Y$ and $Q_Y$ is defined as
\begin{align}\label{chi-divergence-def}
D_{\mathrm{log}}(P_Y ||Q_Y):=\int_{y \in \mathcal Y} p(y)\ \log\left(\frac{p(y)}{q(y)}\right)\ \mathrm{d}\ \!y, 
\end{align}
$p$ and $q$ are the probability densities of $P_Y$ and $Q_Y$, respectively. 
\end{definition}
Notice that the KL-divergence in \eqref{epsilon-Markov-def} can be also equivalently expressed as
\begin{align}
&D_{\mathrm{log}}(P_{Y,X,Z} || P_{Y|X} P_{Z|X} P_X) \nonumber\\
=&\mathbb E_X [ D_{\mathrm{log}}(P_{Y,Z|X} || P_{Y|X} P_{Z|X})] \nonumber\\
=&\mathbb E_{X,Z} [ D_{\mathrm{log}}(P_{Y|X,Z} || P_{Y|X}) ],
\end{align}

\begin{lemma}\label{lemma2}
The following assertions are true:

(a) If $Y_t \overset{\epsilon}\leftrightarrow \mathbf X^l_{t-\mu} \overset{\epsilon}\leftrightarrow \mathbf X^l_{t-\mu-\nu}$ is an $\epsilon$-Markov chain for every $\mu, \nu \geq 0$, then the $L$-conditional entropy is given by
\begin{align}\label{eMarkov}
H_L(Y_t|\mathbf X^l_{t-\delta})= g_1(\delta)+O(\epsilon^2)
\end{align}
where $g_1(\delta)$ is a non-decreasing function of $\delta$, given by
\begin{align}\label{g12function}
\!\!g_1(\delta)=&H_L(Y_t | \mathbf X^l_t) + \sum_{k=0}^{\delta-1}~I_L(Y_t;\mathbf X^l_{t-k}  | \mathbf X^l_{t-k-1}),
\end{align}
the $L$-conditional mutual information $I_L(Y; X|Z)$ between $Y$ and $X$ given $Z$ is  
\begin{align}\label{CMI}
I_L(Y; X|Z)=H_L(Y | Z)-H_L(Y|X, Z).
\end{align}

(b) Given $\delta \geq 0$, $H_L(Y_t|\mathbf X^{l}_{t-\delta})$ is a non-increasing function of feature length $l$, i.e., for all $1 \leq l_1 \leq l_2$, 
\begin{align}
H_L(Y_t|\mathbf X^{l_1}_{t-\delta}) \geq H_L(Y_t|\mathbf X^{l_2}_{t-\delta}).
\end{align}
\end{lemma}

Lemma \ref{lemma2} was introduced in our earlier work \cite{ShisherMobihoc, shisher2023learning} for discrete random variables. To ensure completeness of the paper, we restate Lemma \ref{lemma2} for continuous random variables.

\begin{figure*}
  \centering
\begin{subfigure}[b]{0.75\columnwidth}
\includegraphics[width=\linewidth]{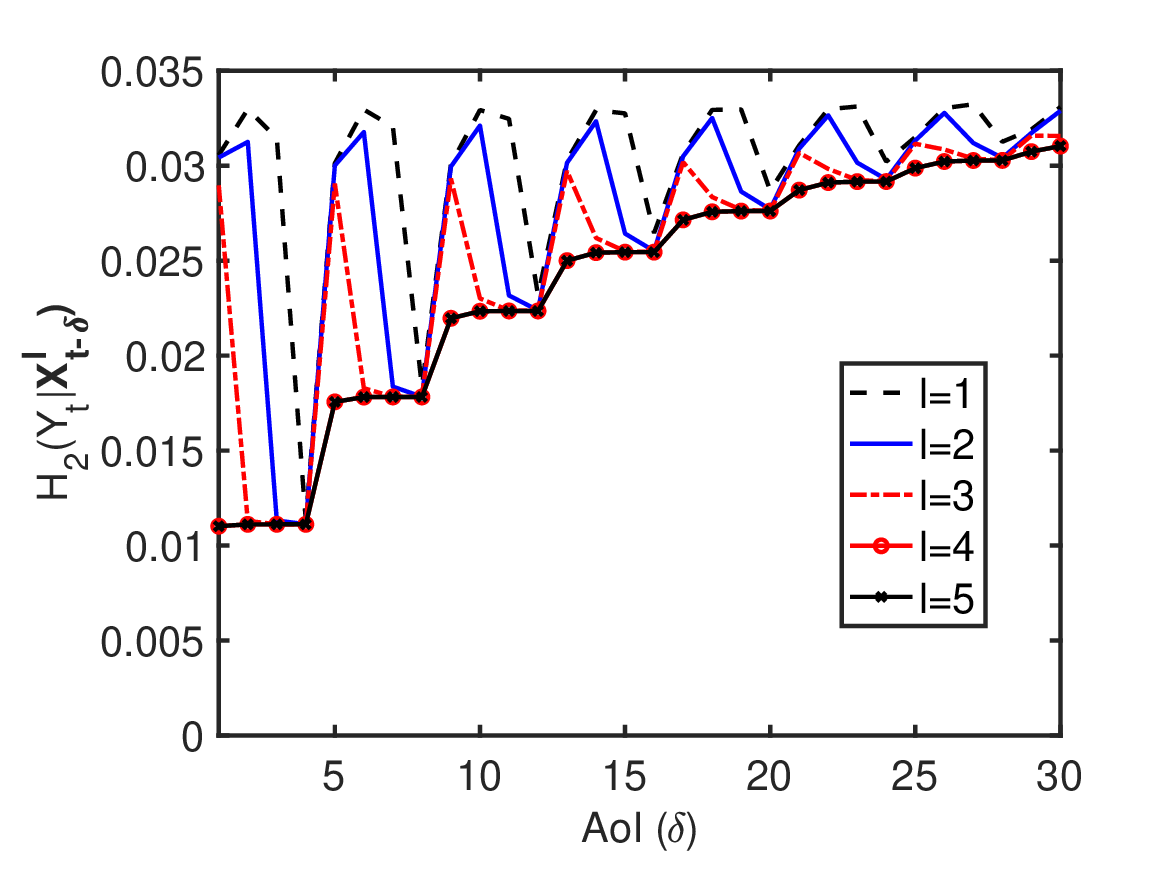}
  \subcaption{\small $H_{\mathrm{2}}(Y_t|\mathbf X^l_{t-\delta})$ vs. AoI ($\delta$)}
\end{subfigure}
 \hspace{20mm}
\begin{subfigure}[b]{0.75\columnwidth}
\includegraphics[width=\linewidth]
{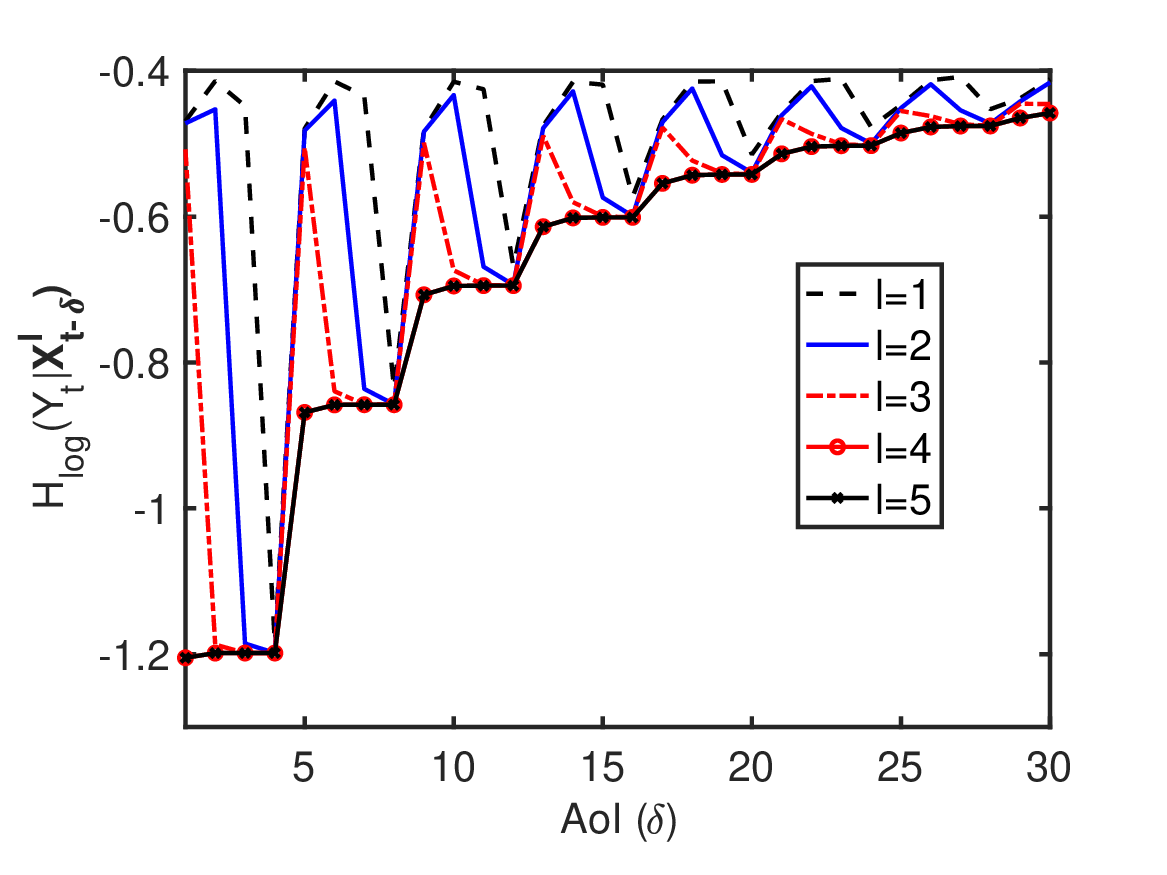}
  \subcaption{\small $H_{\mathrm{log}}(Y_t|\mathbf X^l_{t-\delta})$ vs. AoI ($\delta$)}
\end{subfigure}
%
\caption{$L$-conditional entropy vs. AoI with (a) quadratic loss function  and (b) log loss function (base 2). The $L$-conditional entropy is not always a monotonic function of AoI. An AR($4$) model as discussed in Section \ref{simulation} is considered for this simulation.}
\label{fig:AR}
\end{figure*}

According to Lemma \ref{lemma2}(a), the monotonicity of $H_L(Y_t|\mathbf X^l_{t-\delta})$ in $\delta$ is characterized by the parameter $\epsilon \geq 0$ in the $\epsilon$-Markov chain model. If $\epsilon$ is close to zero, then $Y_t \overset{\epsilon}\leftrightarrow \mathbf X^l_{t-\mu} \overset{\epsilon}\leftrightarrow \mathbf X^l_{t-\mu-\nu}$ is close to a Markov chain, and $H_L(Y_t|\mathbf X^l_{t-\delta})$ is non-decreasing in AoI $\delta$. If $\epsilon$ is large, then $Y_t \overset{\epsilon}\leftrightarrow \mathbf X^l_{t-\mu} \overset{\epsilon}\leftrightarrow \mathbf X^l_{t-\mu-\nu}$ is far from a Markov chain, and $H_L(Y_t|\mathbf X^l_{t-\delta})$ could be non-monotonic in AoI $\delta$. 

Lemma \ref{lemma2}(b) states that $H_L(Y_t|\mathbf X^l_{t-\delta})$ decreases with increasing feature length $l$. A longer feature sequence adds more information that results in better estimation. Nevertheless, increasing the feature length also increases data size, necessitating more communication resources. For example, a longer feature sequence may require a longer transmission time and may end up being stale when delivered, thus resulting in worse inference performance. Recently, a study \cite{shisher2023learning} has investigated a learning and communications co-design problem that jointly optimizes the timeliness and length of feature sequences.

\section{Characterizing the Parameter $\epsilon$ of An $\epsilon$-Markov Chain}
In this section, we show how to evaluate the value of the parameter $\epsilon$ from an AR(p) process. We also analyzed the impact of feature length $l$ on the parameter $\epsilon$.

The parameter $\epsilon$ in $Y_t \overset{\epsilon}\leftrightarrow \mathbf X^l_{t-\mu} \overset{\epsilon}\leftrightarrow \mathbf X^l_{t-\mu-\nu}$ depends on $\mu, \nu,$ and $l$. We denote $\epsilon_{\mu, \nu}(l)$ as the minimum value of $\epsilon$ for which $Y_t \overset{\epsilon}\leftrightarrow \mathbf X^l_{t-\mu} \overset{\epsilon}\leftrightarrow \mathbf X^l_{t-\mu-\nu}$ is an $\epsilon$-Markov chain. By using Definition \ref{epsilonMarkovChain}, we have 
\begin{align}\label{emn}
\epsilon_{\mu, \nu}(l)=\sqrt{I_{\mathrm{log}}(Y_t;\mathbf X^l_{t-\mu-\nu}|\mathbf X^l_{t-\mu})}.
\end{align}

We also denote $\epsilon(l)$ as the minimum value of $\epsilon$ for which $Y_t \overset{\epsilon}\leftrightarrow \mathbf X^l_{t-\mu} \overset{\epsilon}\leftrightarrow \mathbf X^l_{t-\mu-\nu}$ is an $\epsilon$-Markov chain for all $\mu, \nu \geq 0$. Then, we can write
\begin{align}
\epsilon(l)=\max_{\mu, \nu \geq 0} \epsilon_{\mu, \nu}(l).
\end{align}

\begin{proposition}\label{prop3}
    
The following assertions are true for the Gaussian AR($p$) model defined in \eqref{Y}-\eqref{V}.

(a) The minimum value of $\epsilon$ for which the data sequence $(Y_t, \mathbf X^l_{t-\mu}, \mathbf X^l_{t-\mu-\nu})$ satisfies an $\epsilon$-Markov chain property, i.e., $Y_t \overset{\epsilon}\leftrightarrow \mathbf X^l_{t-\mu} \overset{\epsilon}\leftrightarrow \mathbf X^l_{t-\mu-\nu}$, for all $\mu, \nu \geq 0$ is given by 
\begin{align}\label{epsilonl}
\epsilon(l)=\max_{\mu, \nu \geq 0} \epsilon_{\mu, \nu}(l),
\end{align}
where $\epsilon_{\mu, \nu}(l)$ is determined by
\begin{align}\label{epsilonmunu}
\epsilon_{\mu, \nu}(l)= \sqrt{\frac{1}{2} \mathrm{log}\left(\frac{\mathrm{det}(\mathbf R_{[\mathbf X^l_{t-\mu-\nu}, \mathbf X^l_{t-\mu}]})\mathrm{det}(\mathbf R_{[Y_t, \mathbf X^l_{t-\mu}]})}{\mathrm{det}(\mathbf R_{\mathbf X^l_{t-\mu}})\mathrm{det}(\mathbf R_{[Y_t, \mathbf X^l_{t-\mu-\nu}, \mathbf X^l_{t-\mu}]})}\right)},
\end{align} 
$\mathrm{det}(\mathbf A)$ denotes the determinant of a square matrix $\mathbf A$, and $\mathbf R_{\mathbf X}=\mathbb E[\mathbf X^T \mathbf X]$ is the auto-correlation matrix of a random vector $\mathbf X$. 

(b) If $l \geq p$, then
\begin{align}
\epsilon(l)=0.
\end{align}
\end{proposition}
\ifreport
\begin{proof}
See Appendix \ref{prop3p}
\end{proof}
\fi

In Proposition \ref{prop3}(a), we present a closed-form expression for computing the parameter $\epsilon(l)$. Utilizing Proposition \ref{prop3}(a), one can derive $\epsilon(l)$ from the auto-correlation function of a data sequence generated from the AR(p) model. Proposition \ref{prop3}(b) implies that if the feature length $l$ is greater than or equal to the order $p$ of the AR($p$) model, then $\epsilon(l)$ equals $0$. By integrating Proposition \ref{prop3}(b) with Lemma \ref{lemma2}, we can conclude that if the feature length $l$ is greater than or equal to the order $p$, the $L$-conditional entropy becomes a non-decreasing function of AoI. However, transmitting longer features demands more communication resources \cite{shisher2023learning}.

\section{Numerical Results}\label{simulation}

In this section, we utilize Proposition \ref{prop1}-\ref{prop3} to compute the estimation error and the parameter $\epsilon(l)$ for the following AR($4$) process:
\begin{align}\label{sX}
X_{t}&=0.1 X_{t-1}+0.8 X_{t-p}+W_t,\\\label{sY}
Y_t&=X_t+N_t,
\end{align}
where $W_t \in \mathbb R$ and $N_t \in \mathbb R$ are {i.i.d.} Gaussian noises over time with zero mean and variances $0.01$ and $0.001$, respectively. The goal is to estimate $Y_t$ using a feature sequence $\mathbf X^l_{t-\delta}=[X_{t-\delta}, X_{t-\delta-1}, \ldots, X_{t-\delta-l+1}]$. 
\subsection{Evaluating $L$-conditional Entropy Using Propositions \ref{prop1}-\ref{prop2}}
We compute the $L$-conditional entropy of $Y_t$ given $\mathbf X^l_{t-\delta}$ for two different loss functions: (a) quadratic loss and (b) log loss, using \eqref{H_2} and \eqref{H_log}, respectively. In Fig. \ref{fig:AR}, we illustrate the $L$-conditional entropy $H_{2}(Y_t|\mathbf X^l_{t-\delta})$ associated with quadratic loss and the $L$-conditional entropy $H_{\mathrm{log}}(Y_t|\mathbf X^l_{t-\delta})$ associated with log loss (base 2). Both $H_2(Y_t|\mathbf X^l_{t-\delta})$ and $H_{\mathrm{log}}(Y_t|\mathbf X^l_{t-\delta})$ exhibit similar behavior with respect to AoI $\delta$ and feature length $l$, but they are measured in different scale and are used in different applications.     
\subsection{Evaluating $\epsilon(l)$ Using Proposition \ref{prop3}}
We determine $\epsilon(l)$ through the following steps: Firstly, we calculate $\epsilon_{\mu, \nu}(l)$ using \eqref{epsilonmunu} given $\mu, \nu$, and $l$. Subsequently, we compute $\epsilon(l)$ by maximizing $\epsilon_{\mu, \nu}(l)$ over all $\mu, \nu \geq 0$. However, this needs to compute $\epsilon_{\mu, \nu}(l)$ for an infinite number of $\mu$ and $\nu$, which is not possible. We find that when $\mu$ or $\nu$ exceed a large value, $\epsilon_{\mu, \nu}(l)$ becomes either $0$ or close to $0$ for all $l$. Therefore, we can choose an upper bound denoted as $M$ and compute $\epsilon(l)$ by maximizing $\epsilon_{\mu, \nu}(l)$ over all $0\leq \mu, \nu \leq M$. In our simulation, we set $M=50$. The outcomes of $\epsilon(l)$ for feature length $l=1, 2, 3, 4, 5$ are presented in Table \ref{table:1}.

\begin{table}[h!]
\centering
\resizebox{8cm}{!}{
\begin{tabular}{ |c|c|c|c|c|c| } 
\hline
Feature length $l$ & 1 & 2 & 3 & 4 & 5  \\
\hline
$\epsilon(l)$ & 1.55 & 1.49 & 1.39 & 0 & 0 \\ 
\hline
\end{tabular}}
\caption{$\epsilon(l)$ for $l=1, 2, 3, 4, 5$.}\label{table:1}
\end{table}

\subsection{Analysis of the Numerical Results}
Fig. \ref{fig:AR} and Table \ref{table:1} illustrate that as the feature length $l$ increases, the parameter $\epsilon(l)$ tends to zero, and the $L$-conditional entropy becomes a monotonic function of AoI $\delta$. Specifically, when the feature length $l$ reaches the order $p$ of the AR($p$) process, the parameter $\epsilon(l)$ equals zero and hence, the $L$-conditional entropy becomes a monotonic function of AoI. Moreover, as the feature length $l$ increases, the $L$-conditional entropy reduces. However, beyond the order $p$, further increases in feature length do not result in the reduction of the $L$-conditional entropy. It is evident from Fig. \ref{fig:AR} that the $L$-conditional entropy for $l=4$ and $l=5$ remains the same for the AR($4$) model.

\section{Conclusion}
This paper investigates the impact of information freshness on the remote estimation of AR($p$) processes. Employing a new $\epsilon$-Markov chain model, we demonstrate that the estimation error does not always degrade monotonically as the observations become stale. We provide closed-form expressions for computing both the estimation error and the parameter $\epsilon$ for AR($p$) processes. Both theoretical analyses and numerical results illustrate that, with an increasing feature length, $\epsilon$ converges to zero and the estimation error converges to a non-decreasing function of AoI.

\begin{appendices}
\section{Proof of Lemma \ref{lemma1}}\label{plemma1}
By using \eqref{LY}-\eqref{LYx}, we can obtain from \eqref{instantaneous_err1} that
\begin{align}\label{l1-1}
&\mathrm{err}_{\mathrm{estimation}}(\delta, l)\nonumber\\
&=\min_{\phi \in \Phi} \mathbb E_{Y, \mathbf X^l \sim P_{Y_t, \mathbf X^{l}_{t-\delta}}}\bigg[L\bigg(Y,\phi\bigg(\mathbf X^l, \delta\bigg)\bigg)\bigg]\nonumber\\
&=\mathbb E_{\mathbf x^l\sim P_{\mathbf X^l_{t-\delta}}}\left[ \min_{\phi(\mathbf \mathbf x^l, \delta) \in \mathcal A} \mathbb E_{Y \sim P_{Y_t|\mathbf X^{l}_{t-\delta}=\mathbf x^l}}\bigg[L\bigg(Y,\phi\bigg(\mathbf x^l, \delta\bigg)\bigg)\bigg]\right]\nonumber\\
&=\mathbb E_{\mathbf x^l\sim P_{\mathbf X^l_{t-\delta}}}\left[ \min_{a \in \mathcal A} \mathbb E_{Y \sim P_{Y_t|\mathbf X^{l}_{t-\delta}=\mathbf x^l}}\bigg[L(Y,a)\bigg]\right]\nonumber\\
&=\mathbb E_{\mathbf x^l\sim P_{\mathbf X^l_{t-\delta}}}\left[H_L(Y_t|\mathbf X^l_{t-\delta}=\mathbf x^l)\right]\nonumber\\
&=H_L(Y_t|\mathbf X^l_{t-\delta}),
\end{align}
where the second equality holds because $\Phi$ contains all functions that map from $\mathbb R^l \times \mathbb Z^{+}$ to $\mathcal A$.

\section{Proof of Proposition \ref{prop1}}\label{prop1p}
We begin with the definitions of $L$-\emph{divergence} and $L$-\emph{mutual information}. The $L$-divergence $D_L(P_{Y} || Q_{Y})$ of $P_{Y}$ from $Q_{Y}$ can be expressed as \cite{ShisherMobihoc, Dawid2004, farnia2016minimax}
\begin{align}\label{divergence}
&D_L(P_{Y} || Q_{Y}) \nonumber\\
 &=\!\mathbb E_{Y \sim P_{Y}}\left[L\left(Y, a_{P_{Y}}\right)\right]-\mathbb E_{Y \sim P_{Y}}\left[L\left(Y, a_{Q_{Y}}\right)\right],
\end{align}
where $a_{P_Y}$ is the optimal solution to 
\begin{align}
\min_{a \in \mathcal A} \mathbb E_{Y\sim P_Y}[L(Y, a)].
\end{align}
The \emph{$L$-mutual information} $I_L(Y;X)$ is defined as \cite{ShisherMobihoc,Dawid2004, farnia2016minimax}
\begin{align}\label{MI}
I_L(Y; X)=& \mathbb E_{X \sim P_{X}}\left[D_L\left(P_{Y|X}||P_{Y}\right)\right]\nonumber\\
=&H_L(Y)-H_L(Y|X) \geq 0,
\end{align}
which measures the performance gain in estimating $Y$ by observing $X$. The $L$-conditional mutual information $I_L(Y; X | Z)$ is given by 
\begin{align}\label{CMI}
I_L(Y; X|Z)=& \mathbb E_{X, Z \sim P_{X, Z}}\left[D_L\left(P_{Y|X, Z}||P_{Y | Z}\right)\right]\nonumber\\
=&H_L(Y | Z)-H_L(Y|X, Z) \geq 0.
\end{align}

Using \eqref{MI}, the $L$-conditional entropy $H_{\mathrm{log}}(Y_t|\mathbf X^l_{t-\delta})$ associated with log loss can be expressed as 
\begin{align}\label{p1}
H_{\mathrm{log}}(Y_t|\mathbf X^l_{t-\delta})=H_{\mathrm{log}}(Y_t)-I_{\mathrm{log}}(Y_t; \mathbf X^l_{t-\delta}).
\end{align} 
For jointly Gaussian random vectors $\mathbf Y \in \mathbb R^m$ and $\mathbf X \in \mathbb R^n$, we can obtain \cite{polyanskiy2014lecture}
\begin{align}\label{I_log}
I_{\mathrm{log}}(\mathbf Y;\mathbf X)=\frac{1}{2} \mathrm{log} \frac{\mathrm {det}(\Sigma_{\mathbf X})\mathrm {det}(\Sigma_{\mathbf Y})}{\mathrm{det}(\Sigma_{[\mathbf X, \mathbf Y]})},
\end{align}
where $\Sigma_{\mathbf X}:=\mathbb E[(\mathbf X-\mathbb E[\mathbf X])]^T\mathbb E[(\mathbf X-\mathbb E[\mathbf X])]$ denotes the covariance matrix of the row vector $\mathbf X$. If $\mathbb E[\mathbf X]=0$, then $\Sigma_{\mathbf X}=\mathbf R_{\mathbf X}$. 

By using $\mathbb E[Y_t]=0$, $\mathbb E[\mathbf X^l_{t-\delta}]=0$, \eqref{HYlog}, \eqref{p1}, and \eqref{I_log}, we obtain \eqref{H_log}.

\ifreport
\section{Proof of Proposition \ref{prop2}}\label{prop2p}
The conditional expectation of $Y_t$ given $\mathbf X^l_{t-\delta}=\mathbf x^l$, i.e., $\mathbb E[Y_t|\mathbf X^l_{t-\delta}=\mathbf x^l]$ is the optimal estimator of 
\begin{align}
\min_{\phi(\mathbf x^l, \delta) \in \mathbb R} \mathbb E_{Y \sim P_{Y_t|\mathbf X^{l}_{t-\delta}=x^l}}\bigg[(Y,\phi(\mathbf x^l, \delta))^2\bigg].
\end{align}
By substituting $L(y, \phi(\mathbf x^l, \delta))=(y-\phi(\mathbf x^l, \delta))^2$ and $\phi(\mathbf x^l, \delta)=\mathbb E[Y_t|\mathbf X^l_{t-\delta}=\mathbf x^l]$ into \eqref{l1-1}, we obtain
\begin{align}\label{p2-1}
H_2(Y_t|\mathbf X^l_{t-\delta})=\mathbb E[(Y_t-\mathbb E[Y_t|\mathbf X^l_{t-\delta}])^2].
\end{align}

Since $Y_t$ and $\mathbf X^l_{t-\delta}$ are jointly Gaussian with $\mathbb E[Y_t]=0$, $\mathbb E[\mathbf X^l_{t-\delta}]=0$, and
\begin{align}
\mathbb E[Y_t X_{t-k}]=\mathbb E[X_t X_{t-k}]+\mathbb E[N_t X_{t-k}]
=\mathbb E[X_t X_{t-k}],
\end{align}
we get \cite[Chapter 7.3]{papoulis2002probability}
\begin{align}\label{p2-2}
\mathbb E[Y_t|\mathbf X^l_{t-\delta}=\mathbf x^l]=\mathbf A(\mathbf x^l)^T,
\end{align}
where 
\begin{align}\label{p2-3}
\mathbf A=\mathbb E[X_t \mathbf X^l_{t-\delta}] (\mathbf R_{\mathbf X^l_{t}})^{-1}.
\end{align}
By using orthogonality principle \cite[Chapter 7.3]{papoulis2002probability}, we get 
\begin{align}\label{p2-5}
    \mathbb E[(Y_t-\mathbf A(\mathbf X^l_{t-\delta})^T)\mathbf X^l_{t-\delta}]=0 
\end{align} 
Now, using \eqref{p2-2} and \eqref{p2-5}, we obtain from \eqref{p2-1} that 
\begin{align}\label{p2-4}
H_2(Y_t|\mathbf X^l_{t-\delta})&=\mathbb E[(Y_t-\mathbb E[Y_t|\mathbf X^l_{t-\delta}])^2]\nonumber\\
&=\mathbb E[(Y_t-\mathbf A(\mathbf X^l_{t-\delta})^T)Y_t]\nonumber\\
&=\mathbb E[Y_t^2]-\mathbf A (\mathbb E[Y_t \mathbf X^l_{t-\delta}])^T\nonumber\\
&=\mathbb E[Y_t^2]-\mathbf A (\mathbb E[X_t \mathbf X^l_{t-\delta}])^T.
\end{align}
By substituting \eqref{p2-3} into \eqref{p2-4}, we obtain \eqref{H_2}.

\ignore{\section{Proof of Theorem \ref{theorem1}}\label{ptheorem1}
To prove Theorem \ref{theorem1}, we need the following lemma.
\begin{lemma}[\textbf{$\epsilon$-data processing inequality}] \label{Lemma_CMI}
If $Y \overset{\epsilon}\leftrightarrow X \overset{\epsilon}\leftrightarrow Z$ is an $\epsilon$-Markov chain, then 
\begin{align}
I_L(Y;Z|X)=O(\epsilon).
\end{align}
\end{lemma}
Lemma 1 is proved in \cite{shisherMobihoc, ShisherToN} for discrete random variables with finite state space. In this paper, we prove it for continuous random variables.

We have 
\begin{align}\label{CMI1}
&I_L(Y ; Z | X) \nonumber\\
=& \int_{x \in \mathcal X, z \in \mathcal Z} P_{X,Z}(x,z)D_L(P_{Y|X=z,Z=z} || P_{Y|X=z}) \mathrm{d}\ \!x\ \mathrm{d}\ \!z\,
\end{align}
Using Pinsker's inequality, we get 
\begin{align}\label{Pinsker}
    &\mathrm{TV}(P_{Y|X=x, Z=z}, P_{Y|X=x})\nonumber\\ 
    =&\frac{1}{2} \int_{y \in \mathcal Y}|P_{Y|X=x, Z=z}(y)-P_{Y|X=x}(y)|  \mathrm{d}\ \!y\ \nonumber\\
    &\leq \sqrt{\frac{1}{2 \mathrm{log e}}D_{\mathrm{log}}(P_{Y|X=x, Z=z}||P_{Y|X=x})}.
\end{align}
Because $Y \overset{\epsilon}\leftrightarrow X \overset{\epsilon}\leftrightarrow Z$ is an $\epsilon$-Markov chain, we have 
\begin{align}
\int_{x \in \mathcal X, z \in \mathcal Z} D_{\mathrm{log}}(P_{Y|X=x, Z=z}||P_{Y|X=x}) \mathrm{d}\ \!x\ \mathrm{d}\ \!z\ \leq \epsilon^2,
\end{align}
which along with \eqref{Pinsker} yields
\begin{align}
\int_{x \in \mathcal X, z \in \mathcal Z}\mathrm{TV}(P_{Y|X=x, Z=z}, P_{Y|X=x}) \mathrm{d}\ \!x\ \mathrm{d}\ \!z\ =O(\epsilon).
\end{align}
Finally, by using Taylor series expansion of \eqref{CMI1}, we have 
\begin{align}
&I_L(Y ; Z | X)\nonumber\\
&=O\left(\int_{x \in \mathcal X, z \in \mathcal Z}\mathrm{TV}(P_{Y|X=x, Z=z}, P_{Y|X=x}) \mathrm{d}\ \!x\ \mathrm{d}\ \!z\right)\nonumber\\
&=O(\epsilon).
\end{align}}

\section{Proof of Proposition \ref{prop3}}\label{prop3p}
Part (a): The Shannon's conditional mutual information $I_{\mathrm{log}}(Y_t;\mathbf X^l_{t-\mu-\nu}| \mathbf X^l_{t-\mu})$ can be derived as follows: 
\begin{align}\label{p3-1}
&I_{\mathrm{log}}(Y_t;\mathbf X^l_{t-\mu-\nu}| \mathbf X^l_{t-\mu})\nonumber\\
\overset{(a)}{=}& H_{\mathrm{log}}(Y_t|\mathbf X^l_{t-\mu})-H_{\mathrm{log}}(Y_t|\mathbf X^l_{t-\mu}, \mathbf X^l_{t-\mu-\nu})\nonumber\\
=&H_{\mathrm{log}}(Y_t)-H_{\mathrm{log}}(Y_t|\mathbf X^l_{t-\mu}, \mathbf X^l_{t-\mu-\nu})\nonumber\\
&-H_{\mathrm{log}}(Y_t)+H_{\mathrm{log}}(Y_t|\mathbf X^l_{t-\mu})\nonumber\\
\overset{(b)}{=}&I_{\mathrm{log}}(Y_t;\mathbf X^l_{t-\mu-\nu},\mathbf X^l_{t-\mu})-I_{\mathrm{log}}(Y_t;\mathbf X^l_{t-\mu})\nonumber\\
\overset{(c)}{=}&\frac{1}{2}\mathrm{log}\left(\frac{\mathrm{det}(\Sigma_{[\mathbf X^l_{t-\mu-\nu}, \mathbf X^l_{t-\mu}]})\mathrm{det}(\Sigma_{Y_t})}{\mathrm{det}(\Sigma_{[Y_t, \mathbf X^l_{t-\mu-\nu}, \mathbf X^l_{t-\mu})}]}\right)\nonumber\\
&-\frac{1}{2}\mathrm{log}\left(\frac{\mathrm{det}(\Sigma_{\mathbf X^l_{t-\mu}})\mathrm{det}(\Sigma_{Y_t})}{\mathrm{det}(\Sigma_{[Y_t,\mathbf X^l_{t-\mu}]})}\right)\nonumber\\
=&\frac{1}{2} \mathrm{log}\left(\frac{\mathrm{det}(\Sigma_{[\mathbf X^l_{t-\mu-\nu}, \mathbf X^l_{t-\mu}]})\mathrm{det}(\Sigma_{[Y_t, \mathbf X^l_{t-\mu}]})}{\mathrm{det}(\Sigma_{\mathbf X^l_{t-\mu}})\mathrm{det}(\Sigma_{[Y_t, \mathbf X^l_{t-\mu-\nu}, \mathbf X^l_{t-\mu}]})}\right)\nonumber\\
\overset{(d)}{=}&\frac{1}{2} \mathrm{log}\left(\frac{\mathrm{det}(\mathbf R_{[\mathbf X^l_{t-\mu-\nu}, \mathbf X^l_{t-\mu}]})\mathrm{det}(\mathbf R_{[Y_t, \mathbf X^l_{t-\mu}]})}{\mathrm{det}(\mathbf R_{\mathbf X^l_{t-\mu}})\mathrm{det}(\mathbf R_{[Y_t, \mathbf X^l_{t-\mu-\nu}, \mathbf X^l_{t-\mu}]})}\right),
\end{align}
where (a), (b), and (c) hold due to \eqref{CMI}, \eqref{MI}, and \eqref{I_log}, respectively and (d) holds because of $\mathbb E[Y_t]=0$ and $\mathbb E[X_{t-k}]=0$ for all $k$.

Now, by substituting \eqref{p3-1} into \eqref{emn}, we get \eqref{epsilonmunu}.

Part (b): If $l \geq p$, by using \eqref{V}, we can express $\mathbf X^l_{t-\mu}$ as a function of $\mathbf X^l_{t-\mu-\nu}$ for any $\mu, \nu \geq 0$. Hence, if $l \geq p$, then
\begin{align}\label{p3b-1}
H_L(Y_t|\mathbf X^l_{t-\mu}, \mathbf X^l_{t-\mu-\nu})=H_L(Y_t|\mathbf X^l_{t-\mu}).
\end{align}
By using \eqref{CMI} and \eqref{p3b-1}, we get that if $l \geq p$, then
\begin{align}\label{p3b-2}
I_{L}(Y_t;\mathbf X^l_{t-\mu-\nu}|\mathbf X^l_{t-\mu})=0.
\end{align}
From \eqref{epsilonmunu} and \eqref{p3b-2}, we obtain that if $l \geq p$, then $\epsilon_{\mu, \nu}(l)=0$ for all $\mu, \nu \geq 0$. Thus, $\epsilon(l)=0$. This concludes the proof.
\else

\fi
\end{appendices}
\bibliographystyle{IEEEtran}
\bibliography{refshisher}
\end{document}